\documentclass[]{article}

\usepackage{ifdraft}
\ifdraft{\newcommand{\authnote}[3]{{\color{#3} {\bf  #1:} #2}}}{\newcommand{\authnote}[3]{}}

\usepackage[utf8]{inputenc}
\usepackage{amsmath}
\usepackage{latexsym}
\usepackage{amsfonts}
\usepackage{fullpage,amsfonts,bm}
\usepackage{amsmath,amssymb,amsthm,mathtools}
\usepackage{mleftright}\mleftright{}

\usepackage[pagebackref,final]{hyperref}
\usepackage{enumitem}
\usepackage{algpseudocode}
\usepackage{float}

\newtheorem{theorem}{Theorem}

\newtheorem{lemma}[theorem]{Lemma}

\newtheorem*{definition*}{Definition}

\newcommand{\R}{\mathbb{R}}
\newcommand{\C}{\mathbb{C}}

\newcommand{\F}{\mathbb{F}_p}
\newcommand{\Fn}{{\mathbb{F}_p^n}}
\newcommand{\Prf}[2]{\Pr_{f\in_R F_{#1}} \left[ {#2} \right] }
\newcommand{\PrfD}[1]{\Prf{D}{#1} }
\newcommand{\midpoint}{M} 

\newcommand{\spa}{\mbox{\rm span}}

\newcommand{\dom}{\mbox{\rm dom}}

\renewcommand{\i}{\mathbf i}

\newcommand{\eps}{\varepsilon}
\newcommand{\nrm}[1]{\left\lVert#1\right\rVert}
\def\01{\{0,1\}}

\usepackage[usenames,dvipsnames,svgnames,table]{xcolor}

\title{Simon's problem for linear functions}
\author{Joran van Apeldoorn\thanks{QuSoft, CWI, the Netherlands. Both authors are supported by the Netherlands Organization for Scientific Research, grant number 617.001.351. The first author is also partially supported by QuantERA, project QuantAlgo 680-91-034. {\tt \{apeldoor,gribling\}@cwi.nl}}
  \and
  Sander Gribling\footnotemark[1]
}
\date{}

\begin{document}
\maketitle

\begin{abstract}
  Simon's problem asks the following: determine if a function $f: \{0,1\}^n \rightarrow \{0,1\}^n$ is one-to-one or if there exists a unique $s \in \{0,1\}^n$ such that $f(x) = f(x \oplus s)$ for all $x \in \{0,1\}^n$, given the promise that exactly one of the two holds. A classical algorithm that can solve this problem for every $f$ requires~$2^{\Omega(n)}$ queries to $f$. Simon~\cite{simon:power} showed that there is a quantum algorithm that can solve this promise problem for every $f$ using only $\mathcal O(n)$ quantum queries to $f$. A matching lower bound on the number of quantum queries was given in~\cite{knp:simonJ}, even for functions $f: \Fn \to \Fn$. We give a short proof that $\mathcal O(n)$ quantum queries is optimal even when we are additionally promised that $f$ is linear. This is somewhat surprising because for linear functions there even exists a \emph{classical} $n$-query algorithm.
\end{abstract}

\section{Introduction}

In 1994, Simon~\cite{simon:power} showed the existence of a query problem where quantum algorithms offer an exponential improvement over the best randomized classical algorithms that have a bounded error probability of, say, at most 1/3. The problem he considers is the following:

\medskip

\noindent \emph{Given a function $f: \{0,1\}^n \rightarrow \{0,1\}^n$ with the promise that it either (1) is one-to-one or (2) admits a unique $s \in \{0,1\}^n$ such that $f(x) = f(x \oplus s)$ for all $x \in \{0,1\}^n$, decide which of the two holds.}

\medskip

\noindent Simon showed that there is a quantum algorithm which can solve this promise problem for any $f$ using $\mathcal O(n)$ quantum queries to $f$, i.e., using $\mathcal O(n)$ applications of the unitary $|x\rangle |b\rangle \mapsto |x\rangle|b\oplus f(x)\rangle$.\footnote{In fact, Simon considered the problem of finding the non-zero string $s$, if it exists. Here we focus on the decision version of his problem. However, all upper bounds mentioned are derived from algorithms which also find $s$.} This offers an exponential improvement over classical algorithms, since Simon also showed that at least $2^{\Omega(n)}$ classical queries of the form $x \mapsto f(x)$ are needed in order to succeed with probability at least $2/3$. The question we are interested in is the optimality of Simon's quantum algorithm and its generalization to finite fields. Let $p$ be a prime power and let $\F$ be the finite field with $p$ elements. Simon's problem over $\F$ can be formulated as follows:

\medskip

\noindent \emph{Given a function $f: \F^n \rightarrow \F^n$ with the promise that it either (1) is one-to-one or (2) admits a one-dimensional subspace $H \subset \F^n$ such that for all $x,y \in \F^n$, $f(x) = f(y) \Leftrightarrow x-y \in H$, decide which of the two holds.}

\medskip

Koiran et al.~\cite{knp:simonJ} (for an earlier version see~\cite{knp:simon}) showed that the quantum query complexity of Simon's problem over $\F$ is $\Theta(n)$.\footnote{They even prove the analogous lower bound for the hidden subgroup problem over Abelian groups, see Section~\ref{sec:concl}.} Here we show that the lower bound of $\Omega(n)$ quantum queries holds even when $f$ is additionally promised to be linear. That is, a quantum algorithm which can solve Simon's problem over~$\F$ for any linear function requires $\Omega(n)$ quantum queries to $f$. Interestingly, this shows that for the class of linear functions there is no quantum advantage: classically, one can also fully determine a linear function using $n$ queries, by querying a basis.
\begin{definition*}[Linear Simon's problem]
  Given a linear function $f: \Fn \to \Fn$, with the promise that either $|\ker(f)| = 1$ or $|\ker(f)|=p$, decide which of the two holds.
\end{definition*}
Our main result (proved in Section~\ref{sec:proofs}) is the following.
\begin{theorem} \label{thm:mainLB}
  Let $\mathcal A$ be a $T$-query quantum algorithm for the Linear Simon's problem with success probability at least $2/3$. Then $T = \Omega(n)$.
\end{theorem}

We follow the same proof structure as~\cite{knp:simonJ}, using the polynomial method~\cite{bbcmw:polynomialsj}. More specifically, we show that, averaged over a subset of functions, the acceptance probability of a $T$-query quantum algorithm is a polynomial of degree at most $2T$ in the size of the kernel. We then obtain the lower bound by appealing to~\cite[Lemma~5]{knp:simonJ} which states that any polynomial with the correct success probabilities has degree $\Omega(n)$.
However, where~\cite{knp:simonJ} average over all functions, we only consider linear functions over $\F^n$. Surprisingly this simplifies the proof substantially. We also give a slightly simplified proof of~\cite[Lemma~5]{knp:simonJ}.

\paragraph{Notation}

For a set $K \subseteq \F^n$ we call $s:K \rightarrow \F^n$ a \emph{partial} function and we say that $f:\F^n \rightarrow \F^n$ \emph{extends} $s$ if $f(x) = s(x)$ for all $x \in K$. We write $s \preceq f$ if $f$ extends $s$. Let $S_k$ be the set of all partial functions defined on a domain of size at most $k$. Let $\deg_x(f)$ be the degree of $f$ as a polynomial in the variable $x$.
We define $F=\{f : \Fn \rightarrow \Fn \ | \ f \text{ linear} \}$ as the set of all linear functions from $\Fn$ to $\Fn$. For each $k \in \{0,1,\ldots,n\}$ and $D = p^k$ we let $F_D$ be the subset of $F$ consisting of linear functions whose kernel has size $D$, i.e., $F_D=\{f \in F \mid |\ker(f)| = D  \}$.
Finally, we use $\i^2 = -1$ and we use square brackets $[\cdot]: \{\mathrm{true,false}\} \to \{0,1\}$ to denote the function that maps true to $1$ and false to $0$.

\section{Proof of Theorem~\ref{thm:mainLB}} \label{sec:proofs}
The proof of Theorem~\ref{thm:mainLB} is based on a well-known method of lower bounding the quantum query complexity of a Boolean function $G:\{0,1\}^n \to \{0,1\}$: the polynomial method introduced by Beals et al.~\cite{bbcmw:polynomialsj}. Let us first sketch the polynomial method in the setting of their paper.
A $T$-query quantum algorithm $\mathcal A$ for computing $G(x)$ (for every $x \in \{0,1\}^n$) can be described by a Hilbert space $\C^n \otimes \C^2 \otimes \C^m$, a sequence of $T$ unitary matrices $U_0,\ldots, U_T$ acting on the space, and an oracle $O_x$ that is defined as
\[
O_x: |i\rangle |b \rangle |w\rangle \mapsto |i \rangle |b \oplus x_i \rangle |w\rangle.
\]
The definition of the oracle explains the tensor product structure of the Hilbert space $\C^n \otimes \C^2 \otimes \C^m$: the first part corresponds with a query input, the second with a query output, and the last with extra work space.
The quantum algorithm then works as follows. It starts in a fixed state, say $|0 \rangle |0 \rangle |0\rangle$, and then alternates between applying the unitaries and queries before deciding on its output via a measurement to the second register of the final state. Concretely, the state of the algorithm before the final measurement is as follows:
\[
U_T O_x U_{T-1} O_x \cdots O_x U_1 O_x U_0 |0 \rangle |0 \rangle |0 \rangle =: \sum_{(i,b,w) \in [n] \times  \{0,1\} \times [m] } \alpha_{i,b,w}(x) |i\rangle |b \rangle |w \rangle
\]
where $\alpha_{i,b,w}(x) \in \C$. The crucial observation is that the amplitudes $\alpha_{i,b,w}(x)$ of the final state are  polynomials in the input variables $x_i$ of degree at most $T$. Indeed, applying the oracle to, e.g., a state $\alpha |i \rangle |0 \rangle |w\rangle + \beta |i \rangle |1\rangle |w\rangle$ leads to the state
\[
\big((1-x_i) \alpha + x_i  \beta\big) |i\rangle |0\rangle |w\rangle + \big(x_i \alpha + (1-x_i) \beta\big) |i\rangle |1\rangle |w\rangle.
\]
This shows that applying the oracle once increases the degree by at most 1. Since the unitaries do not depend on $x$ and are linear transformations, they do not increase the degree. Instead of viewing the amplitudes as polynomials in the variables $x_i$, it will be more convenient to think of them as homogeneous (degree $T$) polynomials in the Kronecker delta variables $\delta_{x_i,1}:= x_i$ and $\delta_{x_i,0} := (1-x_i)$. The probability of measuring a $1$ in the second register of the final state, i.e., the acceptance probability $P(x)$, is then given by the sum of the squared amplitudes of states with a $1$ in the second register:
\[
P(x) = \sum_{i \in [n],w \in [m]} |\alpha_{i,1,w}(x)|^2 = \sum_{\substack{s\subseteq [n]\times \{0,1\}\\|s| \leq  2T}} \beta_s \prod_{(i,b)\in s} \delta_{x_i,b}
\]
where the real numbers $\beta_s$ are the coefficients of the monomials $\prod_{(i,b)\in s} \delta_{x_i,b}$ in $P(x)$. If $\mathcal A$ computes $G$ with high success probability, then $P(x)$ will be close to $G(x)$ for every $x \in \{0,1\}^n$ which may be used to prove a degree lower bound on $P(x)$.  However, proving lower bounds on the degree of $P(x)$ directly is often complicated. A common technique is to average $P(x)$ over multiple inputs in order to reduce the problem to studying a univariate polynomial. For example, for a symmetric\footnote{A Boolean function $G$ is symmetric if $G(x)$ only depends on the Hamming weight $|x|$ of $x$.} function $G: \{0,1\}^n \rightarrow \{0,1\}$ averaging $P(x)$ over all permutations of $n$ elements reduces the problem to studying univariate polynomials $q(|x|)$ which approximate $G(x)$ (for which tight degree bounds are known)~\cite{bbcmw:polynomialsj}.

\medskip

The above version of the polynomial method is easily generalized to inputs that are not Boolean (see, e.g.,~\cite{aaronson&shi:collision}). We will do so here for the setting corresponding to the Linear Simon's problem.

Let $\mathcal A$ be a $T$-query algorithm for the Linear Simon's problem and let $P(f)$ be the acceptance probability of $\mathcal A$ on the input $f$. As before, we can write
\[
P(f) = \sum_{\substack{s\subseteq\Fn\times\Fn \\|s| \leq 2T}} \beta_s \prod_{(x,y)\in s} \delta_{f(x),y}.
\]
When we view $s$ as a partial function, this expression can be rewritten in terms of $f$ extending $s$:
\[
P(f) = \sum_{s\in S_{2T}} \beta_s [ s \preceq f ],
\]
where $S_{2T}$ is the set of all partial functions $s$ with $|\dom(s)|\leq 2T$. As above, it will turn out to be useful to average $P(f)$ over all linear functions $f$ with a kernel of size $D$, i.e., we consider the average acceptance probability $Q(D)$ over all functions with a kernel of size $D$:
\[
Q(D) = \sum_{f\in F_D} \frac{1}{|F_D|} P(f) = \sum_{f\in F_D} \frac{1}{|F_D|}\sum_{s\in S} \beta_s [ s \preceq f ] = \sum_{s\in S} \beta_s \frac{1}{|F_D|} \sum_{f\in F_D}  [ s \preceq f ] = \sum_{s\in S} \beta_s Q_s(D).
\]
Here $Q_s(D)$ is the probability that a uniformly random $f \in F_D$ extends $s$:
\[
Q_s(D) = \frac{1}{|F_D|} \sum_{f\in F_D}  [ s \preceq f ] = \PrfD{s\preceq f}
\]
In the next two sections we will prove that the degree of $Q$ needs to be at least linear in $n$, and that the degree of each $Q_s$ (and hence of $Q$) is upper bounded by $2T$. Together these results implies Theorem~\ref{thm:mainLB}.

\subsection{Lower bound on the degree}
For $k \in \{0,1,\dots,n\}$, $Q(p^k)$ represents an acceptance probability and therefore $Q(p^k) \in [0,1]$. Moreover, if the algorithm succeeds with probability at least $2/3$, then $Q(1) \geq 2/3$ and $Q(p) \leq 1/3$. The lemma below shows that such a $Q$ has degree $\Omega(n)$. We give a slightly simplified proof for completeness.

\begin{lemma}[{\cite[Lemma~5]{knp:simonJ}}] \label{lem:deglb}
  For every polynomial $Q$ such that $Q(1)\geq 2/3$, $Q(p) \leq 1/3$ and $Q(p^k) \in [0,1]$ for all $k\in \{0,\dots,n\}$, it holds that $\deg(Q) \geq n/4$.
\end{lemma}
\begin{proof}
  Assume that $Q$ is a polynomial of degree $d\leq n/2$  (otherwise we are done), so that its derivative $Q'$ is of degree $d-1$ and its second derivative $Q''$ is of degree $d-2$. Consider the $2d-2$ intervals of the form $(p^a,p^{a+1})$ where $a = n-(2d-2),\ldots, n-1$. Since together $Q'$ and $Q''$ have at most $2d-3$ roots, there is such an interval for which both polynomials have no roots with real part in it; let $a \geq n-(2d-2)$ be the integer corresponding to this interval and let $M:=\frac{1+p}{2}p^{a}$ be the middle of this interval. By the mean value theorem we know that there is an $x_0 \in [1,p]$ for which $|Q'(x_0)|\geq \frac{1}{3(p-1)}$.
  To show the degree lower bound it suffices to prove the following chain of inequalities:
  \[
  \frac{1}{ p^{2d-2}} \stackrel{(*)}{\leq} \left|\frac{Q'(\midpoint)}{Q'(x_0)} \right| \stackrel{(**)}{\leq} \frac{3(p-1)}{\frac{p-1}{2} p^{n-2d+2}}.
  \]
  Indeed, if the above chain of inequalities holds, then $6 \geq p^{n-4d +4} \geq 2^{n-4d+4}$ which implies that $n-4d+4 \leq 3$, i.e., $d \geq \frac{n+1}{4}$.
  \\$\mathbf{(*)}$ For the lower bound we will use the following elementary fact:
  \begin{equation}\label{eq:simplefact}
    \text{if } 0\leq v < w\text{ and } 0\leq y, \text{ then } \frac{v+y}{w+y} \geq \frac{v}{w}
  \end{equation}
  Denote the roots of $Q'$ by $b_j+c_j \i$, for $j \in [d-1]$. Then $Q'(x) = \lambda \prod_{j=1}^{d-1}(x-b_j-c_j \i)$ for some $\lambda \in \R$ and hence
  \[
  \left|\frac{Q'(\midpoint)}{Q'(x_0)} \right| =  \left| \prod_{j=1}^{d-1} \frac{\midpoint - b_j-c_j \i}{x_0-b_j-c_j \i}  \right| =  \prod_{j=1}^{d-1} \left| \frac{\midpoint - b_j-c_j\i}{x_0-b_j-c_j\i}  \right| = \prod_{j=1}^{d-1} \sqrt{\frac{\left(\midpoint - b_j\right)^2+c_j^2}{\left(x_0-b_j\right)^2 + c_j^2}}
  \]
  We will show that each factor in the product is bounded from below by $1/p^2$.
  Considering the $j$-th factor, if $| x_0 - b_j| \leq |\midpoint - b_j|$ then we are clearly done. Hence, assume $|x_0 - b_j| > |\midpoint - b_j|$, that is, $b_j> \frac{\midpoint - x_0}{2} \geq p^{a-1}$. We use~\eqref{eq:simplefact}:
  \[
  \sqrt{\frac{\left(\midpoint - b_j\right)^2+c_j^2}{\left(x_0-b_j\right)^2 + c_j^2}} \geq \left| \frac{\midpoint - b_j}{x_0-b_j} \right|
  \]
  Since we know that $b_j > p^{a-1}$ and $b_j\not\in (p^a,p^{a+1})$ there are two cases to consider:
  \begin{itemize}
  \item If $b_j \in (p^{a-1},p^a]$, then $\displaystyle{\left| \frac{\midpoint - b_j}{x_0-b_j} \right| \geq \inf_{x\in (p^{a-1},p^a)} \left| \frac{\midpoint - x}{x_0-x} \right| = \left| \frac{\midpoint - p^a}{x_0-p^a} \right| \geq  \frac{1}{2}} \geq \frac{1}{p^2}$

  \item If $b_j \in [p^{a+1},\infty)$, then $\displaystyle{
    \left| \frac{\midpoint - b_j}{x_0-b_j} \right| = \left| \frac{ - \frac{1+p}{2}p^{a} + b_j}{ -x_0+b_j} \right| = \left| \frac{  \frac{p-1}{2}p^a + (b_j - p^{a+1})}{ p^{a+1}-x_0+(b_j-p^{a+1})} \right| \geq \frac{p^{a-1}}{p^{a+1} - x_0} \geq \frac{1}{p^2}}$\\
    where we use~\eqref{eq:simplefact} and $\frac{p-1}{2} \geq \frac{1}{p}$ for the first inequality.
  \end{itemize}

  \noindent $\mathbf{(**)}$ By construction $|Q'(x_0)| \geq \frac{1}{3(p-1)}$, so it remains to show that $|Q'(\midpoint)| \leq (\frac{p-1}{2} p^{n-2d+2})^{-1}$. Assume towards a contradiction that $|Q'(\midpoint)| > (\frac{p-1}{2} p^{a})^{-1}$. Since $Q''$ has no roots with real part in the interval $(p^a,p^{a+1})$, $Q'$ is either strictly increasing or strictly decreasing on the interval $(p^a,p^{a+1})$. Therefore, there is an interval $(\alpha,\beta)$ (with $\alpha,\beta \in \{p^a, \midpoint, p^{a+1}\}$) of length $\frac{p-1}{2} p^a$ where $|Q'(x)| > (\frac{p-1}{2} p^{a})^{-1}$. By the fundamental theorem of calculus this implies that $|Q(\alpha) - Q(\beta)| >1$. This is a contradiction, since we have $1 \geq |Q(p^{a+1}) - Q(p^a)| \geq |Q(\alpha) - Q(\beta)|$, where the last inequality follows by monotonicity of $Q$ on the interval $(p^a,p^{a+1})$. It follows that
  \[
  |Q'(\midpoint)| \leq \left(\frac{p-1}{2} p^{a}\right)^{-1} \leq \left(\frac{p-1}{2} p^{n-2d+2}\right)^{-1}.
  \]
  We conclude that $\displaystyle{\frac{1}{ p^{2d-2}} \leq \left|\frac{Q'(\midpoint)}{Q'(x_0)} \right| \leq \frac{3(p-1)}{\frac{p-1}{2} p^{n-2d+2}}}$ and hence that $d \geq n/4$. \qedhere
\end{proof}

\subsection{Upper bound on the degree}
We now show that the degree of each $Q_s$ is upper bounded by $2T$.
\begin{lemma}\label{lem:upperbound}
Given a partial linear function $s: \dom(s) \rightarrow \Fn$, $\deg_D(Q_s)\leq \dim(\spa(\dom(s)))$.
\end{lemma}
\begin{proof}
  Let $K := \spa(\dom(s))$ and $k:=\dim(K)$. We can extend $s$ uniquely to a linear function on $K$. Define $Z:= \ker(s) \subseteq K$ and $z := \dim(Z)$, and $Y := Z^{\perp} \cap K$. For a function $f:\Fn \rightarrow \Fn$ in $F_D$ we write $H:= \ker(f)$, $h := \dim(H)$ and $D := |H| = p^h$. We show that $\PrfD{s\preceq f}$ has degree at most $k$ as a polynomial in $D$.
We analyze this probability in three parts:
  \begin{align*}
    \PrfD{s\preceq f} &= \PrfD{Z \subseteq H \land Y \cap H = \{ 0 \} } \PrfD{s\preceq f \mid Z \subseteq H \land Y \cap H = \{ 0 \} }\\
    &= \PrfD{Z \subseteq H} \PrfD{Y \cap H = \{ 0 \} \mid Z \subseteq H }  \PrfD{s\preceq f  \mid Z \subseteq H \land Y \cap H = \{ 0 \} }.
  \end{align*}
We show that
\begin{enumerate}
\item[(1)] $\PrfD{Z \subseteq H}$ is a polynomial in $D$ of degree at most $z$,
\item[(2)] $\PrfD{Y \cap H = \{ 0 \} \mid Z \subseteq H }$ is a polynomial in D of degree at most $k-z$,
\item[(3)] $\PrfD{s\preceq f \mid Z \subseteq H \land Y \cap H = \{ 0 \} }$ does not depend on $D$.
\end{enumerate}
Together, this implies that $\PrfD{s\preceq f}$ is a polynomial in $D$ of degree at most $k$.
\begin{enumerate}
\item[(1)] The probability that $Z \subseteq H$ equals the fraction of $h$-dimensional subspaces of $\Fn$ that contain~$Z$. There are $\alpha(n,h) = \prod_{i=0}^{h-1} (p^n - p^i)$ ways to pick $h$ linearly independent vectors in a space of dimension~$n$, and hence there are $\beta(n,h) = \frac{\alpha(n,h)}{\alpha(h,h)}$ different subspaces of dimension $h$ in $\Fn$. The number of $h$-dimensional subspaces that contain $Z$ equals the number of $(h-z)$-dimensional subspaces in an $(n-z)$-dimensional space. Hence
  \[
  \PrfD{Z \subseteq H} = \frac{\beta(n-z,h-z)}{\beta(n,h)} = \prod_{i=0}^{z-1} \frac{p^h - p^i}{p^n-p^i},
  \]
  which is a degree-$z$ polynomial in terms of $D = p^h$.
\item[(2)] We have
$\displaystyle{\PrfD{Y \cap H = \{ 0 \} \mid Z \subseteq H } =   \PrfD{Y/Z \cap H/Z = \{ 0 \}}}$ where $Y/Z$ and $H/Z$ are subspaces of $\Fn/Z \simeq \F^{n-z}$.
By construction we have that $\dim(Y/Z) = \dim(Y) = k-z$, $\dim(H/Z) = h-z$. 
The probability $ \PrfD{Y/Z \cap H/Z = \{ 0 \}}$ equals the number of $(h-z)$-dimensional bases of $\F^{n-z}$ which are linearly independent from $Y$, divided by $\beta(n-z,h-z)$. That is,
\[
\PrfD{Y/Z \cap H/Z = \{ 0 \}} = \frac{\frac{\prod_{i=0}^{h-z-1} p^{n-z}-p^{k-z+i}}{\alpha(h-z,h-z)}}{\frac{\alpha(n-z,h-z)}{\alpha(h-z,h-z)}}
  = \frac{\prod_{i=0}^{h-z-1} p^{n-z}-p^{k-z+i}}{\alpha(n-z,h-z)}  = \frac{\prod_{i=0}^{k-z-1} p^{n-z}-p^{h-z+i}}{\alpha(n-z,k-z)}
\]
where the last equality is obtained using ${\alpha(n-z,h-z) = \alpha(n-z,k-z)  \prod_{i=k-z}^{h-z-1} p^{n-z} - p^i}$.
It follows that $\displaystyle{\PrfD{Y/Z \cap H/Z = \{ 0 \}} = \frac{\prod_{i=0}^{k-z-1} p^{n-z}-p^{h-z+i}}{\alpha(n-z,k-z)}}$ is a polynomial in $D=p^h$ of degree $k-z$. We mention in passing that, alternatively, one can arrive at the same expression by looking at the probability that a random $Y$ is linearly independent from a fixed $H$.

\item[(3)]   Finally we consider $ \PrfD{s\preceq f \mid Z \subseteq H \land Y \cap H = \{ 0 \} } $. Since $Z \subseteq H$, we know that $f$ and $s$ agree on $Z$. Hence, $f$ extends $s$ if their values agree on $Y$. Let $b_1,\dots,b_{k-z}$ be a basis for $Y$, then $f$ and $s$ agree on $Y$ if and only if they agree on $b_1, \ldots, b_{k-z}$.  Since we condition on the event $Y \cap H = \{0\}$, the probability that this happens does not depend on $D = p^h$.
\end{enumerate}
\end{proof}

\section{Open problems} \label{sec:concl}
To conclude, we propose the following open problems:
\begin{itemize}
\item Koiran et al.~\cite{knp:simonJ} lift the lower bound on Simon's problem over $\Fn$ to the hidden subgroup problem over finite Abelian groups:

\medskip

\emph{Given a (finite Abelian) group $G$ and a function $f: G \to X$ with the promise that there is a subgroup $H \leq G$ of rank either $0$ or $1$ (i.e., either trivial, or generated by a single element), such that $f(g) = f(g')$ if and only if $g-g' \in H$, decide which of the two holds.}

\medskip

One recovers Simon's problem over $\Fn$ by taking $G = X = \Fn$. A natural question is whether or not the hidden subgroup problem over finite Abelian groups also remains equally hard when we are additionally promised that $f$ is an endomorphism. The reduction used by Koiran et al.~combined with our result gives a smaller and more structured set of hard instances of the hidden subgroup problem over Abelian groups. However, the functions obtained from this reduction will only be endomorphisms on a subgroup of $G$, not on all of $G$.

\item While the general Simon's problem has no natural extension to $\R^n$, the linear Simon's problem can possibly be extended to $\R^n$. For example: given matrix-vector multiplication queries $x \mapsto Ax$ for a symmetric matrix $A$ with $\nrm{A}\leq 1$, decide if $\lambda_{\min}(A)\leq \eps$ or $\lambda_{\min}(A)\geq 2\eps$. It remains an open question to prove a lower bound on this problem. An $\Omega(n)$ lower bound could have implications for quantum convex optimization. In particular this may resolve an open question posed in recent work~\cite{apeldoor:convexoracles} regarding the number of queries needed to optimize a convex function.

\item Aaronson and Ben-David~\cite{aaronson:sculpt} introduced the idea of \emph{sculpting} functions. They characterized the total Boolean functions for which there is a promise on the input such that restricted to that promise there is an exponential separation between quantum and classical query complexity.
We propose the related idea of \emph{over-sculpting}: bringing the classical query complexity down to the quantum query complexity.
More specifically, for which (possibly partial) Boolean functions $f$ does there exist a promise $P$ such that:
  \[
    Q_{1/3}(f) \leq o(R_{1/3}(f))
  \]
  \[
    Q_{1/3}(f) = \Theta(Q_{1/3}(f|_P)) = \Theta(R_{1/3}(f|_P)).
  \]
  Simon's problem does not correspond to a Boolean function since the input alphabet is not Boolean\footnote{An input for Simons problem is a function $f: \Fn \rightarrow \Fn$, which can be viewed as a string of length $p^n$ over the input alphabet $\Fn$.}, but our results show that Simon's problem can be over-sculpted in this slightly different setting.
\end{itemize}
\paragraph{Acknowledgements} We would like to thank Ronald de Wolf for many helpful comments and discussions. We would also like to thank Andr\'as Gily\'en for useful discussions.

\bibliographystyle{alpha}
\bibliography{qc}
\end{document}